\documentclass[9pt,a4paper]{article}
\usepackage{arxiv}
\usepackage[numbers]{natbib}
\usepackage[utf8]{inputenc} 
\usepackage[T1]{fontenc}    
\usepackage{hyperref}       
\usepackage{url}            
\usepackage{booktabs}       
\usepackage{amsfonts}       
\usepackage{nicefrac}       
\usepackage{microtype}      
\usepackage{lipsum}
\usepackage{kotex}
\usepackage{graphicx}
\usepackage[caption=false]{subfig}
\usepackage{epstopdf}
\usepackage{algorithmic}
\usepackage{tabularx,multirow,booktabs}
\usepackage{amsmath, amsthm, amssymb, amsfonts, amscd, latexsym}
\usepackage{rotating}
\usepackage{pdflscape}
\usepackage{tabulary}

\newcommand{\du}{\textrm{d}}

\newtheorem{theorem}{Theorem}

\title{Global Stability of an SEIR Epidemic Model where Empirical Distribution of Incubation Period has Approximated by Coxian Distribution}

\author{
  Sungchan Kim\footnotemark[1]\thanks{Department of Mathematics, Pusan National University, Busan 46241, South Korea}
   \And
 Jong Hyuk Byun\footnotemark[2]\thanks{Finance·Fishery·Manufacture Industrial Mathematics Center on Big Data, Pusan National University, Busan 46241, South Korea}\\
   \And
 Il Hyo Jung\footnotemark[1] \footnotemark[2] \footnotemark[3]\thanks{Corresponding Author (E-mail: \texttt{ilhjung@pusan.ac.kr})} }
\allowdisplaybreaks

\begin{document}
\maketitle

\begin{abstract}
In this work, we have developed a Coxian distributed SEIR model in incorporating an empirical incubation period. We show  that the global dynamics are completely determined by a basic reproduction number. An application of the Coxian distributed SEIR model using data of an empirical incubation period is explored. The model may be useful for resolving causing the realistic intrinsic parts in classical epidemic models since Coxian distribution approximately converges to any distribution.
\end{abstract}


\maketitle

\section{Introduction}
Compartmental models in epidemiology are well used to simplify  mathematical modeling of infectious diseases (\cite{anderson1992infectious}). Its origin is in the early 20th century, with an important early work being that of Kermack and McKendrick in 1927 (\cite{kermack1927contribution}). The population is subdivided into three groups: susceptible, infectious and recovered groups, and the model is referred to as an SIR epidemic model. The model is applicable to infectious diseases such as measles, chickenpox, or mumps, where infection confers immunity.

For many important infections there is a significant incubation period during which the individual has been infected but is not yet infectious themselves (\cite{smith2001global}). During this period the individual is in compartment $E$, known as the exposed term, and such type models are called SEIR models. To express the distribution of incubation period, we consider age-structured SEIR framework. Let $S(t)$, $E(t)$, $I(t)$ and $R(t)$ be the fraction of susceptible, exposed, infectious and recovered population at time $t$, respectively. Motivated by \cite{hethcote1980integral}, consider an SEIR model of integro-differential type:
\begin{subequations}\label{master}
\begin{align}
\cfrac{\du S(t)}{\du t}&= \mu -\beta(t)S(t)I(t)-\mu S(t), \label{master1}\\
E(t)&=E_0P(t)\exp(-\mu t)+\int_{0}^t \beta(t-u) S(t-u) I(t-u) P(u) \exp(-\mu u) \du u,\label{master2}\\
\cfrac{\du I(t)}{\du t}&= \int_0^t \beta(t-u)S(t-u)I(t-u)[-P_u(u)]\exp(-\mu u)\du u -(\gamma+\mu) I(t),\label{master3}\\
\cfrac{\du R(t)}{\du t}&= \gamma I(t)-\mu R(t),\label{master4}
\end{align}
\end{subequations}
where $\beta(t)$ is the time-dependent contact rate, $1/\gamma$ is the mean period of infectiousness, $P(u)$ is the probability of remaining incubated (assumed asymptomatic and not infectious) to time $u$ from entering, and $\mu$ is the natural mortality. $P_u$ denotes the derivative with respect to $u$. The initial conditions are given as $S(0)=S_0,~E(0)=E_0,~I(0)=I_0,~R(0)=R_0\textrm{, and } S_0+E_0+I_0+R_0=1$. In general, it is reasonable to assume that a function $P:[0,\infty)\rightarrow[0,1] \in \mathcal{L}^1$ is Lebesgue integrable, differentiable a.e., non-increasing, $P(0)=1$, and $\lim_{t\rightarrow \infty} P(t) = 0$ (\cite{van2007modeling}). Especially, let $T$ be a continuous random variable of incubation period of some infectious diseases. If we know the distribution of incubation period heuristically, we can define $P$ by
\begin{align}\label{survival}
P(u)=\mathbb{P}(T>u)= 1-F(u),
\end{align}where $F$ is a cumulative density function of random variable of infectious period $T$. $P$ is called the {\it survival function}. Many previous studies have developed with the various survival functions $P$ such as Dirac-delta \cite{mittler1998influence, macdonald2008biological}, Exponential \cite{li1995global}, Gamma \cite{lloyd2001realistic,krylova2013effects,safi2011qualitative,feng2007epidemiological}, Mittag-Leffler \cite{angstmann2017fractional,byun2019modeling}, and joint \cite{melesse2010global,iwami2010global} types. In some cases, this framework of SEIR model can have been simplified by system of ordinary or fractional differential equations with various incubation period distribution (see \autoref{tablehistory}).

Following these, in this work, we construct an SEIR epidemic model, which is named  Coxian distributed SEIR model, to incorporate empirical incubation period distribution. Coxian distribution, one of the phase-type distributions, can be considered as a mixture of hypo-- and hyper--exponential distributions. The novelty of the distribution is the density in the class of all non--negative distribution functions \cite{buchholz2014input} and so {\it all types of incubation periods are approximated to Coxian distribution}. A basic reproduction number, $R_0$, was derived. Applying a Lyapunov theory, we show that the basic reproduction number determines the global stability of equilibrium of our model with constant contact rate case. We also give an application illustrating how to use a Coxian distributed SEIR model with empirical incubation period data. The limitation of the model is also discussed in the last section.

\section{Derivation of Coxian Distributed SEIR Model}
A Coxian distribution with $n$ phases is defined as the time until absorption into state $0$, starting from state $n$, of the Markov process with discrete states in continuous time. A focal state, $X$, can be partitioned into $n$ sub-states $X_n,\cdots,X_1$, each with independent dwell time distributions that are exponentially distributed with rates $\lambda_i$, $i=n,n-1,\cdots,1$. Inflow rates into the state $X_n$ can be described by non-negative, integrable inflow rates ${I}_X(t)$. Particles that transition out of a sub-state $X_k$ at time $t$ transition into either a different sub-state $X_{k-1}$ with probability $\bar{p}_{k-1}$, or enter the recipient state $X_0$ with probability $p_{k-1}=1-\bar{p}_{k-1}$, for $k=n,n-1,\cdots,1$. Then, Coxian distributed random variable $X$ \cite{buchholz2014input} is the time until absorption in $X_0$ starting from state $n$ and the probability density function $f_X$ is given as $f_X(t)=\mathbf{p}\exp{(t\mathbf{Q})}\mathbf{q}$ where 
\begin{align}\label{matrixform}\begin{split}
\mathbf{p}&=\begin{bmatrix}1&0&\cdots&0\end{bmatrix}\in M_{1\times n}(\mathbb{R}),
\\
\mathbf{Q}&=\begin{bmatrix}
-\lambda_n&\bar{p}_{n-1}\lambda_n & 0 & &\cdots &0 \\
0&-\lambda_{n-1}&\bar{p}_{n-2}\lambda_{n-1}&0 &\cdots &0\\
&&&\vdots&&\\
0&\cdots &&&-\lambda_2&\bar{p}_1\lambda_2\\
0&\cdots &&&0&-\lambda_1
\end{bmatrix}\in M_{n\times n}(\mathbb{R}),\\
\mathbf{q}&=-\mathbf{Q1}=\begin{bmatrix}
p_{n-1}\lambda_n &p_{n-2}\lambda_{n-1}&\cdots &p_1\lambda_2 &\lambda_1
\end{bmatrix}^\top\in M_{n\times 1}(\mathbb{R}),\end{split}\end{align}
  $\mathbf{1}$ is an $n\times 1$ vector of ones and $\mathbf{Q}$ is transition rate matrix, $0<p_i<1$ for $i=1,2,\cdots,n-1$ and $\
  \lambda_i>0$ for $i=1,2,\cdots,n$. 
 The survival function $P$ is given as 
 \begin{align}\label{survcox}
 P(t)=\mathbf{p}\exp({t\mathbf{Q}})\mathbf{1}.
 \end{align} 
The Laplace transform $\mathcal{L}$ of $f_X$ is given explicitly by
 \begin{align}\label{laplace}
 [\mathcal{L}f_X](s)=\mathbf{p}(s\mathbf{I}-\mathbf{Q})^{-1}\mathbf{q} =\sum_{i=0}^{n-1} \left[ a_i \prod_{j=1}^{n-i} \cfrac{\lambda_{n+1-j}}{s+\lambda_{n+1-j}}\right],\end{align} where $\mathbf{I}$ is $n\times n$ identity matrix and $a_i=p_i\prod_{j=1}^{n-1-i}\bar{p}_{j+i}$. 
 
Derivation of Coxian distributed SEIR model is the following: assume that infectious period is distributed by Coxian with the survival function \eqref{survcox} in the model \eqref{model}. First, note that if we put the $1\times n$ vector $\mathbf{P}$ as:
\begin{align*}
\mathbf{P}(t)&:=\mathbf{p}\exp(t\mathbf{Q})=\begin{bmatrix}
P_n&P_{n-1}&\cdots&P_2&P_1
\end{bmatrix},
\end{align*}
where $P_i$'s are functions of $t$, then the survival function $P$ is expressed as
\begin{align}\label{survcox2}P(t)=\sum_{i=1}^n P_n(t).
\end{align}
Since $\frac{\du  [\exp{(t\mathbf{Q})}]}{\du t}= \exp{(t\mathbf{Q})}\mathbf{Q}$, $\mathbf{P}'(t)=\mathbf{P}(t)\mathbf{Q}$ holds and so \begin{align} \label{dif}
\begin{split}
\cfrac{\du P_n}{\du t} &= -\lambda_n P_n,\\
\cfrac{\du P_i}{\du t} &= \bar{p}_{i}\lambda_{i+1}P_{i+1}-\lambda_{i}P_i,\textrm{ for } i=n-1,\cdots,1,
\end{split}
\end{align} and \begin{align}\label{pprime}\cfrac{\du P(t)}{\du t}=\sum_{i=1}^n\cfrac{\du P_i(t)}{\du t}=-\lambda_1 P_1(t)-\sum_{i=1}^{n-1}p_i \lambda_{i+1}P_{i+1}(t).\end{align}
Now, if we put \eqref{survcox} to \eqref{master2}, \eqref{master2} has the form from \eqref{survcox2}:
\[	E(t)=\sum_{i=1}^n \int_{0}^\infty P_i(u) \beta(t-u) S(t-u) I(t-u)  \exp(-\mu u) \du u.\] So if we put 
\begin{align}\label{diff} E_i(t)= \int_{0}^\infty P_i(u) \beta(t-u) S(t-u) I(t-u)  \exp(-\mu u) \du u,
\end{align} for $i=n,n-1,\cdots,1$ and differentiating \eqref{diff} for each $i=n,n-1,\cdots,1$, then \eqref{dif} yields:
\begin{align}\label{gm}
\begin{split}
\cfrac{\du E_n(t)}{\du t}&=\beta(t) S(t)I(t) - (\lambda_n+\mu) E_n(t),\\
\cfrac{\du E_{i}(t)}{\du t}&=\bar{p}_{i}\lambda_{i+1} E_{i+1}(t)-(\lambda_{i}+\mu)E_{i}(t),~i=n-1,\cdots,2,1,
\end{split}
\end{align}
and substituting \eqref{pprime} into \eqref{master3} yields
\[\cfrac{\du I(t)}{\du t}= \int_0^t \beta(t-u)S(t-u)I(t-u)\left[\lambda_1 P_1(u)+ \sum_{i=1}^{n-1}p_i\lambda_{i+1}P_{i+1}(u)\right]\exp(-\mu u)\du u -(\gamma+\mu) I(t),\]
and so the \eqref{master3} becomes
\begin{align}\label{mo2}
\cfrac{\du I(t)}{\du t}=\lambda_1 E_1(t)+\sum_{i=1}^{n-1} p_{i}\lambda_{i+1}  E_{i+1}(t) - (\gamma+\mu)I(t).
\end{align}
From \eqref{gm} and \eqref{mo2}, targeted $n-$chained Coxian distributed SEIR model is derived from \eqref{master} when we put $p_0=1$, as follows:
\begin{align}\label{model}
\begin{split}
\cfrac{\du S(t)}{\du t}&=\mu-\beta(t) S(t)I(t)-\mu S(t), \\
\cfrac{\du E_n(t)}{\du t}&=\beta(t) S(t)I(t) - (\lambda_n+\mu)E_n(t),\\
\cfrac{\du E_{i}(t)}{\du t}&=\bar{p}_{i}\lambda_{i+1} E_{i+1}(t)-(\lambda_{i}+\mu)E_{i}(t),~i=n-1,\cdots,2,1,\\
\cfrac{\du I(t)}{\du t}&=\sum_{i=0}^{n-1} p_{i} \lambda_{i+1}  E_{i+1}(t) - (\gamma+\mu)I(t), \\
\cfrac{\du R(t)}{\du t}&=\gamma I(t)-\mu R(t),
\end{split}
\end{align}with $\sum_{i=1}^nE_i(t)=E(t)$. The schematic diagram for the model \eqref{model} is depicted in \autoref{fig1}.

Basic properties of the model are below: the region \[\Omega=\left\{(S,E_n,E_{n-1},\cdots,E_1,I,R)\in\mathbb{R}^{n+3}_{+}:S+\sum_{j=1}^n E_j+I+R\equiv 1 \right\}.\] is obviously positive-invariant. Also if the initial data $S(0), E_i(0), I(0)$ and $R(0)$ for $i=n,n-1,\cdots,1$ are positive, then the solutions $S(t), E_i(t), I(t)$ and $R(t)$ of the model \eqref{model} are non-negative for all $t>0$. Indeed, if we assume that there exists $t_1>0$ such that $S(t)>0,E_i(t)>0,I(t)>0,R(t)>0$, $i=n,n-1,\cdots,1$ in $t \in [0,t_1)$ and $S(t_1)\cdot\prod_{i=1}^n E_i(t_1) \cdot I(t_1) \cdot R(t_1)=0$ (i.e., one of states is 0 at the time $t_1$). Then, 
\[\cfrac{\du S(t)}{\du t}=\mu-(\zeta(t)+\mu) S(t)\ge -(\zeta(t)+\mu)S(t) \] when we take a force of infection $\zeta(t)=\beta(t)I(t)$, we may get \[S(t_1)\ge S(0) \exp\left\{-\mu t_1-\int_0^{t_1}\zeta(u)\du u\right\}>0\]holds from Gronwall. Similarly, it can be shown that $E_i(t_1)>0,I(t_1)>0$ and $R(t_1)>0$ for $i=n,n-1,\cdots,1$ and this yields a contradiction.

\section{The Basic Reproduction Number, $R_0$}
When the parameters are constant, we can compute the basic reproduction number. Assume the number of infectious human is rare in the early phase, i.e. $S\approx 1$. Consider $\beta(t)\equiv 
\bar{\beta}$. Define the average time that an individual remains in the exposed class becoming infectious without dying, $\hat{P}$, as
$$\hat{P}= \int_0^\infty \exp(-\mu u)P(u)\du u.$$ We note that $\mu \hat{P}$ represents the probability that an exposed individual will die during the course of incubating, and so the probability of surviving the exposed class is $1-\mu \hat{P}$. Note that since $0 \le P(t)\le 1$ for $t\ge 0$, 
the inequality $$0\le \hat{P} \le  \int_0^\infty \exp(-\mu u)\du u = 1/\mu$$ holds and so the probability $\mu\hat{P}$ is well-defined. From the definition of the basic reproduction number, we can represent the basic reproduction number, denoted by $R_0$, as $$R_0=\bar{\beta} \cdot (1-\mu \hat{P})\cdot \frac{1}{\gamma+\mu}.$$ Since $\mu \hat{P} = 1+\int_0^\infty \exp{(-\mu u)}P'(u)\du u$ 
from integration by parts, and the relation $P'(t)=-f_X(t)$ in the usual sense,
 we get $$1-\mu\hat{P} = \int_0^\infty \exp(-\mu u)f_X(u)\du u = [\mathcal{L}f_X](\mu).$$
From \eqref{laplace}, $R_0$ is explicitly represented as:
\begin{align}\label{r0}
R_0=\cfrac{\bar{\beta}}{\gamma+\mu} \sum_{i=0}^{n-1} \left[ a_i \prod_{k=1}^{n-i} \cfrac{\lambda_{n+1-k}}{\lambda_{n+1-k}+\mu}\right],\end{align}
where $a_i=p_i\prod_{j=1}^{n-1-i}\bar{p}_{j+i}$. Note that if we ignore demographic part, i.e. $\mu=0$, then $R_0$ is expressed as $\bar{\beta}/\gamma$, same as the one in classical model.

\section{Equilibria Analysis for the case of constant contact rate}
Consider the case when $\beta(t)=\bar{\beta}$, a positive constant. The model \eqref{model} has a equilibrium 
\begin{align}\mathcal{E}^d=(S^d,E_n^d,E_{n-1}^d, \cdots ,E_1^d,I^d,R^d)=(1,0,0,\cdots,0),\end{align}which is called the disease free equilibrium. For convenience, let \[\Omega_d=\left\{(S,E_n,E_{n-1},\cdots,E_1,I,R)\in\mathbb{R}^{n+3}_{+}:E_n=E_{n-1}=\cdots=E_1=I=R=0 \right\}.\]Now we prove the following properties:
\begin{theorem} \label{thm1}The disease free equilibrium, $\mathcal{E}^d$, of the model \eqref{model} is globally asymptotically stable in the domain $\Omega_d$ if $R_0\le 1$.
\end{theorem}
\begin{proof}Consider the function 
\begin{align*}
V=\sum_{j=1}^n w_j E_j+I, \end{align*}
for positive constants $w_j$, for $j=1,\cdots,n$, where
\begin{align*}
w_j=\sum_{l=1}^j\alpha_l^j\prod_{k=l}^j\cfrac{\lambda_k}{\lambda_k+\mu},
\end{align*}
with \[\alpha_l^j=\begin{cases}p_{j-1}&\textrm{if $l=j$}\\ 
p_{l-1}\prod_{\xi=l}^{j-1}\bar{p}_{\xi}&\textrm{otherwise, i.e. }l=1,2,\cdots,j-1. \end{cases}\]for fixed $j$. Then, we can note that $w_n={(\gamma+\mu)R_0}/{\bar{\beta}}$.
The derivative of $V$ to time $t$ is given by
\begin{align}\label{global1}
\begin{split}
\cfrac{\du V}{\du t}&=\sum_{j=1}^nw_j\cfrac{\du E_j}{\du t}+\cfrac{\du I}{\du t},\\
&=w_n\left[\bar{\beta}S(t)I(t)-(\lambda_n+\mu)E_n(t)\right]+\sum_{j=1}^{n-1}w_j\left[\bar{p}_j\lambda_{j+1}E_{j+1}(t)-(\lambda_j+\mu)E_j(t)\right]\\&~~~+\sum_{j=0}^{n-1}p_j\lambda_{j+1}E_{j+1}(t)-(\gamma+\mu)I(t),\\
&=w_n\bar{\beta}S(t)I(t)-(\gamma+\mu)I(t)+\sum_{j=2}^{n}\left[-w_j(\lambda_j+\mu)+w_{j-1}\bar{p}_{j-1}\lambda_j+p_{j-1}\lambda_j\right]E_j(t)\\&~~~+[p_0\lambda_1-w_1(\lambda_1+\mu)]E_1(t),
\end{split}
\end{align}by rearrangement. Moreover, since 
\begin{align}\label{global2}p_0\lambda_1-w_1(\lambda_1+\mu)=p_0\lambda_1-p_0\lambda_1=0,\end{align}
and
\begin{align*}
w_j&=\sum_{l=1}^j a_l^j \prod_{k=l}^j \cfrac{\lambda_k}{\lambda_k+\mu},\\
&=a_1^j\prod_{k=1}^j\cfrac{\lambda_k}{\lambda_k+\mu}+a_2^j\prod_{k=2}^j\cfrac{\lambda_k}{\lambda_k+\mu}+\cdots+a_j^j\cfrac{\lambda_j}{\lambda_j+\mu},\\
&=\bar{p}_{j-1}\cfrac{\lambda_j}{\lambda_j+\mu}\underbrace{\left[a_1^{j-1}\prod_{k=1}^{j-1}\cfrac{\lambda_k}{\lambda_k+\mu}+a_2^{j-1}\prod_{k=2}^{j-1}\cfrac{\lambda_k}{\lambda_k+\mu}+\cdots+a_{j-1}^{j-1}\cfrac{\lambda_{j-1}}{\lambda_{j-1}+\mu}\right]}_{=\sum_{l=1}^{j-1}a_{l}^{j-1}\prod_{k=l}^{j-1}\frac{\lambda_k}{\lambda_k+\mu}=w_{j-1}}\\&~~~+a_j^j\cfrac{\lambda_j}{\lambda_j+\mu}, \textrm{~~~since $a_k^{j}=\bar{p}_{j-1}\cdot a_{k}^{j-1}$,}\\
&=\cfrac{\lambda_j}{\lambda_j+\mu}\left[\bar{p}_{j-1}w_{j-1} +p_{j-1}\right],
\end{align*}and so we can get the recurrence relation: 
\begin{align}\label{global3} w_j(\lambda_j+\mu)=w_{j-1}\bar{p}_{j-1}\lambda_j+p_{j-1}\lambda_j,
\end{align}
for each $j=2,3,\cdots,n$. Thus, if we substitute \eqref{global2} and \eqref{global3} into \eqref{global1}, we may get
\begin{align*}
\cfrac{\du V}{\du t}&=w_n\bar{\beta}SI-(\gamma+\mu)I+\sum_{j=2}^{n}\underbrace{\left[-w_j(\lambda_j+\mu)+w_{j-1}\bar{p}_{j-1}\lambda_j+p_{j-1}\lambda_j\right]}_{\text{$=0$ by \eqref{global3}}}E_j\\&~~~+\underbrace{[p_0\lambda_1-w_1(\lambda_1+\mu)]}_{\text{=0 by \eqref{global2}}}E_1,\\
&=w_n \bar{\beta}SI-(\gamma+\mu)I,\\
&\le w_n\bar{\beta}I-(\gamma+\mu)I \textrm{ ~~~since $S\le1$~(equality holds iff $S=1$)},\\
&=(\gamma+\mu)\left( \underbrace{w_n\cfrac{\bar{\beta}}{\gamma+\mu}}_{=R_0}-1 \right)I,\\
&=(\gamma+\mu)(R_0-1)I\le 0 \textrm{~~~ if $R_0\le1$}.
\end{align*}
Therefore, $V'\le 0$ if $R_0\le 1$ and $V'=0$ if and only if $E_n=E_{n-1}=\cdots=E_1=I=0$. Also, substituting $I=0$ into the equations for $S$ and $R$ shows that $S\rightarrow 1$ and $R\rightarrow 0$ as $t\rightarrow \infty$. Hence, $V$ is the Lyapunov function in the domain $\Omega_d$. From the LaSalle's invariance principle, the disease free equilibrium $\mathcal{E}^d$ is globally asymptotically stable in $\Omega_d$ if $R_0\le 1$. 
\end{proof}

Next we consider the endemic equilibrium, $\mathcal{E}^\ast=(S^\ast,E_n^\ast,\cdots,E_1^\ast,I^\ast,R^\ast)$, whose components are all positive. Solving the algebraic system in \eqref{model} at fixed point gives
\begin{subequations}
\begin{align}
\bar{\beta}S^\ast I^\ast &= \mu(1-S^\ast), \label{al1}\\
\bar{\beta}S^\ast I^\ast &= (\lambda_n+\mu)E_n^\ast,\label{al2} \\
E_i^\ast &=\cfrac{1}{\lambda_i+\mu}\,\bar{p}_i\lambda_{i+1}E_{i+1}^\ast\textrm{~~~for $i=n-1,\cdots,2,1$},\label{al3}\\
I^\ast&=\cfrac{1}{\gamma+\mu}\sum_{i=0}^{n-1}p_i\lambda_{i+1}E_{i+1}^\ast,=\cfrac{1}{\gamma+\mu}\,\bar{\beta} S^\ast I^\ast (1-\mu \hat{P})\label{al4}\\
R^\ast&=\cfrac{\gamma}{\mu} \,I^\ast.\label{al5}
\end{align}
\end{subequations} From \eqref{al1} and \eqref{al2}, 
\begin{subequations}\label{globale}
\begin{align}\label{globale1}E_n^\ast=\mu(1-S^\ast)\,\cfrac{1}{\lambda_n+\mu},\end{align} and from \eqref{al3}, we get the $E_i^\ast$ for $i=n-1,n-2,\cdots,1$ recursively by 
\begin{align}\label{globale2}E_j^\ast=\mu(1-S^\ast)\, \cfrac{1}{\lambda_j+\mu} \prod_{k=j+1}^n \cfrac{\bar{p}_{k-1}\lambda_k}{\lambda_k+\mu},\textrm{~~~for } j=n-1,\cdots,1.\end{align}
\end{subequations} Substituting \eqref{globale} into \eqref{al4} yields
\begin{align}
\begin{split}
I^\ast&=\mu(1-S^\ast)\,\cfrac{1}{\gamma+\mu}\sum_{i=0}^{n-1}p_i\lambda_{i+1}E_{i+1}^\ast,\\
&=\mu(1-S^\ast)\,\cfrac{1}{\gamma+\mu}\left[ p_0\cfrac{\lambda_1}{\lambda_1+\mu}\prod_{k=2}^n\cfrac{\bar{p}_{k-1}\lambda_k}{\lambda_k+\mu} + p_1\cfrac{\lambda_2}{\lambda_2+\mu}\prod_{k=3}^n\cfrac{\bar{p}_{k-1}\lambda_k}{\lambda_k+\mu}+\cdots + p_{n-1}\cfrac{\lambda_n}{\lambda_n+\mu}\right],\\
&=\mu(1-S^\ast)\,\cfrac{1}{\gamma+\mu}\,\sum_{i=0}^{n-1}\left[a_i \prod_{k=i+1}^n \cfrac{\lambda_k}{\lambda_k+\mu}\right],\\
&=\mu(1-S^\ast)\,\cfrac{1}{\gamma+\mu}\,\sum_{i=0}^{n-1} \left[ a_i \prod_{k=1}^{n-i} \cfrac{\lambda_{n+1-k}}{\lambda_{n+1-k}+\mu}\right]\\
&=\cfrac{\mu}{\bar{\beta}}\,(1-S^\ast)\,R_0.\label{globalee}
\end{split}
\end{align}Moreover, if we substitute \eqref{globalee} into \eqref{al1}, then 
\begin{align*}
I^\ast=\cfrac{\mu}{\bar{\beta}}\,(R_0-1),
\end{align*}
 since $S^\ast \neq 1$. Hence $I^\ast>0$ if $R_0>1$. Also, if $I^\ast>0$, then $0<S^\ast=1/R_0<1$ and $R^\ast>0$ from \eqref{al5}. Moreover, $E_i$'s are positive since $1-S^\ast>0$. Conversely, if $R_0\le 1$, then the model has no positive equilibrium. These results are summarized below.
 \begin{theorem}
 The endemic equilibrium $\mathcal{E}^\ast$ of the model \eqref{model} exists uniquely when $R_0>1$, and no endemic equilibrium otherwise.
 \end{theorem}
Finally, we claim the following:
\begin{theorem} \label{thm3} The endemic equilibrium, $\mathcal{E}^\ast$, of the model \eqref{model} is globally asymptotically stable in the interior of $\Omega$ if $R_0>1$.
\end{theorem}
\begin{proof}
For convenience, define $$\hat{f}(t):=\int_t^\infty \exp(-\mu u)f_X(u) \,\du u \in (0,1),$$where $f_X(x)=-P'(x)$. Then, $1-\mu\hat{P}=\hat{f}(0)$.
Motivated by \cite{wang2014stability}, consider a Lyapunov function $V_e\equiv V_e(S,I)$ as \[V_e = V_{1e}+V_{2e},\]where  
\begin{align*}
V_{1e}&=\hat{f}(0)\cdot G\left(\cfrac{S}{S^\ast}\right)+I^\ast\cdot G\left(\cfrac{I}{I^\ast}\right),\\
V_{2e}&=\bar{\beta}S^\ast I^\ast \int_0^\infty \hat{f}(u) \cdot G\left(\cfrac{S(t-u)I(t-u)}{S^\ast I^\ast}\right)\du u,
\end{align*}
and $G(x)=x-1-\ln x$. Notice that $V_e=0$ when $S=S^\ast, E_n=E_n^\ast,\cdots, R=R^\ast$ and $V_e>0$ otherwise. Differentiating $V_{1e}$ and $V_{2e}$ with respect to the time $t$ yields
\begin{align}\label{v1e}
\begin{split}
\cfrac{\du V_{1e}}{\du t}&= \hat{f}(0)\left(1-\cfrac{S^\ast}{S(t)}\right)\left[\mu-\bar{\beta}S(t)I(t)-\mu S(t)\right]\\&~~~ +\left(1-\cfrac{I^\ast}{I(t)}\right)\left[\int_0^\infty\bar{\beta}S(t-u)I(t-u)\exp(-\mu u)f_X(u)\,\du u-(\gamma+\mu)I(t)\right],\\
&=-\hat{f}(0) \cdot\mu\,\cfrac{(S(t)-S^\ast)^2}{S(t)}+\hat{f}(0)\cdot\bar{\beta}S^\ast I^\ast -\hat{f}(0)\cdot\bar{\beta}S(t)I(t) -\hat{f}(0)\cdot\bar{\beta} \,\cfrac{{S^\ast}^2}{S(t)}\, I^\ast + \hat{f}(0)\cdot\bar{\beta} S^\ast I(t)\\
&~~~+\int_0^\infty\bar{\beta}S(t-u) I(t-u)\exp(-\mu u)f_X(u)\,\du u -\cfrac{I(t)}{I^\ast}\int\bar{\beta} S^\ast I^\ast\exp(-\mu u)f_X(u) \du u \\
&~~~-\cfrac{I^\ast}{I(t)}\,\int_0^\infty\bar{\beta}S(t-u) I(t-u)\exp(-\mu u)f_X(u)\,\du u+\int_0^\infty \bar{\beta}S^\ast I^\ast \exp(-\mu u)f_X(u)\,\du u,
\end{split}
\end{align}
from the relation \eqref{al1} and \eqref{al4}, and,
\begin{align}\label{v2e}
\begin{split}
\cfrac{\du V_{2e}}{\du t} &= \cfrac{\du}{\du t}\,\,\bar{\beta}S^\ast I^{\ast} \int_0^\infty\hat{f}(u)\cdot G\left(\cfrac{S(t-u)I(t-u)}{S^\ast I^\ast}\right)\,\du u,\\
&=\hat{f}(0)\cdot\bar{\beta}S^\ast I^\ast \cdot G\left(\cfrac{S(t)I(t)}{S^\ast I^\ast}\right)-\bar{\beta}S^\ast I^\ast\int_0^\infty G\left(\cfrac{S(t-u)I(t-u)}{S^\ast I^\ast}\right)\,\du \hat{f}(u) ,\\
&=\hat{f}(0)\cdot\bar{\beta}S^\ast I^\ast \cdot G\left(\cfrac{S(t)I(t)}{S^\ast I^\ast}\right)-\bar{\beta}S^\ast I^\ast\int_0^\infty G\left(\cfrac{S(t-u)I(t-u)}{S^\ast I^\ast}\right)\exp(-\mu u)f_X(u)\,\du u ,\\
&=\hat{f}(0)\cdot\bar{\beta}\left(S(t)I(t)-S^\ast I^\ast-S^\ast I^\ast \cdot \ln \cfrac{S(t)I(t)}{S^\ast I^\ast}\right)\\
&~~~-\int_0^\infty  \left[\bar{\beta}S(t-u)I(t-u)-\bar{\beta}S^\ast I^\ast-\bar{\beta}S^\ast I^\ast\cdot \ln\cfrac{S(t-u)I(t-u)}{S^\ast I^\ast}\right]\exp(-\mu u)f_X(u)\,\du u,\\
&=\hat{f}(0)\cdot\left(\bar{\beta}S(t)I(t)-\bar{\beta}S^\ast I^\ast \cdot \ln \cfrac{S(t)I(t)}{S^\ast I^\ast}\right)\\
&~~~-\int_0^\infty \left[\bar{\beta}S(t-u)I(t-u)-\bar{\beta}S^\ast I^\ast\cdot \ln\cfrac{S(t-u)I(t-u)}{S^\ast I^\ast}\right]\exp(-\mu u)f_X(u)\,\du u.
\end{split}
\end{align}Combining \eqref{v1e} and \eqref{v2e}, we can get
\begin{align*}
\cfrac{\du V_e}{\du t} &= \cfrac{\du V_{1e}}{\du t}+\cfrac{\du V_{2e}}{\du t},\\
&= -\hat{f}(0) \cdot\mu\,\cfrac{(S(t)-S^\ast)^2}{S(t)}-\hat{f}(0)\cdot \bar{\beta}S^\ast I^\ast \left[\cfrac{S^\ast}{S(t)}-2\right]\\&~~~-\int_0^\infty \bar{\beta}S^\ast I^\ast  \left[\cfrac{S(t-u)I(t-u)}{S^\ast I(t)} -\ln\cfrac{S(t-u)I(t-u)}{S(t)I(t)}\right]\exp(-\mu u)f_X(u)\, \du u,\\
&= -\hat{f}(0) \cdot\mu\,\cfrac{(S(t)-S^\ast)^2}{S(t)}-\hat{f}(0)\cdot \bar{\beta}S^\ast I^\ast \cdot G\left(\cfrac{S\ast}{S(t)}\right)+\hat{f}(0)\cdot \bar{\beta}S^\ast I^\ast \left[1-\ln \cfrac{S^\ast}{S(t)} \right]\\
&~~~-\int_0^\infty \bar{\beta}S^\ast I^\ast \cdot G\left(\cfrac{S(t-u)I(t-u)}{S^\ast I(t)}\right)\exp(-\mu u)f_X(u) \,\du u\\
&~~~-\underbrace{\int_0^\infty \bar{\beta}S^\ast I^\ast \left[ 1+\ln\cfrac{S(t-u)I(t-u)}{S^\ast I(t)}-\ln\cfrac{S(t-u)I(t-u)}{S(t)I(t)} \right]\exp(-\mu u)f_X(u)\,\du u,}_{=-\int_0^\infty \bar{\beta}S^\ast I^\ast \exp(-\mu u)f_X(u)[1-\ln\frac{S^\ast}{S(t)}]\,\du u =-\hat{f}(0)\cdot \bar{\beta}S^\ast I^\ast \left[1-\ln \frac{S^\ast}{S(t)}\right],}
\\
&= -\hat{f}(0) \cdot\mu\,\cfrac{(S(t)-S^\ast)^2}{S(t)}-\hat{f}(0)\cdot \bar{\beta}S^\ast I^\ast \cdot G\left(\cfrac{S\ast}{S(t)}\right)\\
&~~~-\int_0^\infty \bar{\beta}S^\ast I^\ast  \cdot G\left(\cfrac{S(t-u)I(t-u)}{S^\ast I(t)}\right) \exp(-\mu u)f_X(u)\,\du u,\\
&\le 0,
\end{align*}since $G(x)\ge 0$ for all $x>0$ and so the integrand of the last term are positive. Note $V_e'=0$ only if at the endemic equilibrium and $V_e'<0$ otherwise. Therefore, by the LaSalle's principle, every solution to the equation in the model \eqref{model} approaches the endemic equilibrium for $R_0>1$. Thus the endemic equilibrium $\mathcal{E}^\ast$ is globally asymptotically stable if $R_0>1$.
\end{proof}
\autoref{thm1} and \autoref{thm3} indicate that the disease could be eliminated by maintaining the basic reproduction number less than unity, and conversely, when $R_0$ is greater than unity, the disease persists in the epidemiological point of view.

\section{An application}
The 2009 epidemic of influenza H1N1 in Canada is investigated to explain the procedure from fitting incubation period to applying to our model. The procedure is summarized: approximating empirical distribution to Coxian distribution, and investigating a Coxian distributed SEIR model.

First, we approximate the empirical distribution of incubation period to Coxian distribution, $1-P$, as denoted in \eqref{survival}. Empirical data of incubation period are captured from \cite{tuite2010estimated} and the data are fit to the Exponential and Coxian, respectively. Unfortunately, there is no criterion for choosing the number of substates $n$. However, Akaike Information Criterion, AIC, gives the relative quality of statistical models for a given set of data \cite{akaike1998information} and so we choose $n$ that makes the AIC smallest. Next, the human case data are fit to Coxian distributed SEIR model. The report \cite{hsieh2010epidemic} is referred to obtaining daily human case data. We normalized the model \eqref{model} as multiplying the human population. Since the duration of epidemic is short related to human lifespan, we ignore the demographic effect, and so $\mu=0$ is assumed. Since the rate of infection-related death is too small to be considered, we don't consider the death from the epidemic. Motivated by \cite{tuite2010estimated}, the mean duration of infectiousness $1/\gamma$ is assumed as $7.1$. We consider the initial time ($t=0$) as April 14 when the epidemic started. Since there are two turning points in April 29 and June 4 in the duration of epidemic, we assume the transmission rate $\beta(t)$ consists of step function:
\[\beta(t)=\begin{cases}\beta_0&\textrm{ on }t\in[0,15),\\
\beta_1&\textrm{ on }t\in[15,53),\\
\beta_2&\textrm{ on }t\in[53,\infty),\\
\end{cases}\] and estimate $\beta(t)$ by fitting between the case data and the cumulative prevalence, $\int_0^t\left\{\sum_{i=0}^{n-1}p_i\lambda_{i+1}E_{i+1}(u)\right\} \du u$, from the model as minimizing the sum of squared residual with the corresponding data.

 Figure \ref{fig2a}(a) shows the result of empirical distribution of incubation period to Coxian distributions of 12-chains, which is determined by AIC(corresponding AIC$=-36.3$). We see that the fit curve explain the distribution of empirical incubation period well. Figure \ref{fig2b}(b) indicates the result of investigating Coxian distributed SEIR model. Well fit values of set of parameters are $\beta_0=0.642$, $\beta_1=0.252$, and $\beta_2=0.131$ and so the basic reproduction number \eqref{r0}, $R_0(=\beta_0/\gamma)$ is $4.561$.
To compare the result with the classical model, we fit the empirical distribution of incubation period to exponential and investigate the classical exponential SEIR model. In this case, AIC value of distribution fit is $-33.3$. Exponential model gives $\beta_0=0.584$, $\beta_1=0.234$ and $\beta_2=0.133$ and so the $R_0$ is $4.146$. From the example, we can see that investigating Coxian distributed SEIR model might give such a nice fit results, and moreover, $R_0$ can vary greatly as considering the realistic distribution. This result supports the fact that common assumption for exponentially distributed incubation period always underestimate the basic reproductive number of infection from onset data and considering a realistic incubation period distribution is important to modeling epidemics (\cite{wearing2005appropriate}).

\section{Discussion and Conclusion}
Many previous works have strongly emphasized that modelers should be cautious for considering the intrinsic facts to classical frameworks when the epidemic models for public health is proposed. In the study, we have derived an SEIR model based on Coxian distribution which approximates the distribution of the incubation period. In mathematical analysis, we proved that global stability of equilibria with respect to the value of $R_0$. The novelty of the model is that all types of distributions of the incubation period are fit with Coxian distribution and in addition, since it has realistic distribution of incubation period, it enables to describe the SEIR model of a particular type of distribution, like bimodal, that is not expressed by the conventional way, such as {\it P. vivax} malaria in temperate regions. Moreover, we may extend the Coxian distributed model not just incubation period, but infectious period. The model has some limitations even if several benefits. First, a number of the parameters are needed to fit empirical data in which they brings computation works. Secondly, loss of biological meaning could be yielded because of going out to absorption state without going through the whole chain of incubation. However, we believe that the model may be applicable for resolving the problem caused by simplicity of the models in some cases when sufficient empirical information of the incubation period is given, as shown in the previous section.


\section*{Availability of data and material}
The datasets analyzed during the current study are available in \cite{tuite2010estimated}.

\section*{Competing interests}
The authors declare that they have no competing interests.
  
\section*{Funding}
Not applicable.

\section*{Author's contributions}
IHJ conceptualized the roles and designed the studies. SK did mathematical analysis of the model and illustrated an application. SK and JHB discussed and analyzed all results.  All members verified the results and wrote and reviewed the manuscript.

\bibliographystyle{manuscript_arxiv}
\bibliography{manuscript_arxiv.bib}
\newpage 

\section*{Figures}
\begin{figure}[h!]
	\centering{
		\includegraphics[width=.85\textwidth]{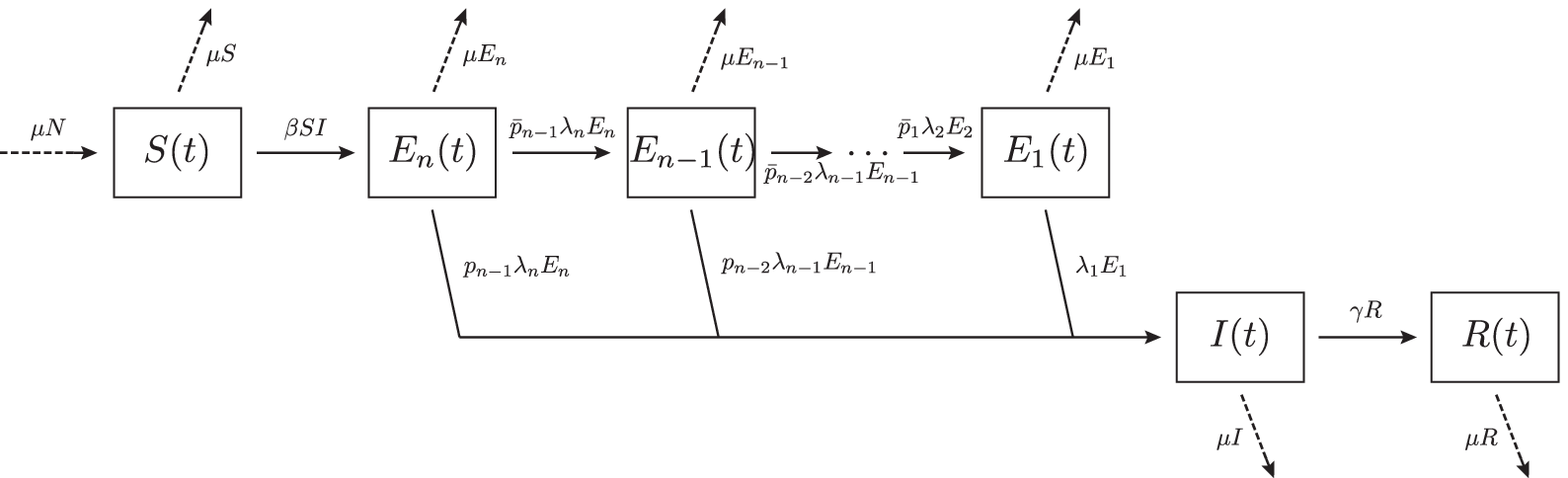}}	\caption{\label{fig1}Scheme of the model \eqref{model}.}
\end{figure}  

\begin{figure}[h!]
\centering 
\subfloat[]{\label{fig2a}\includegraphics[width=.45\textwidth]{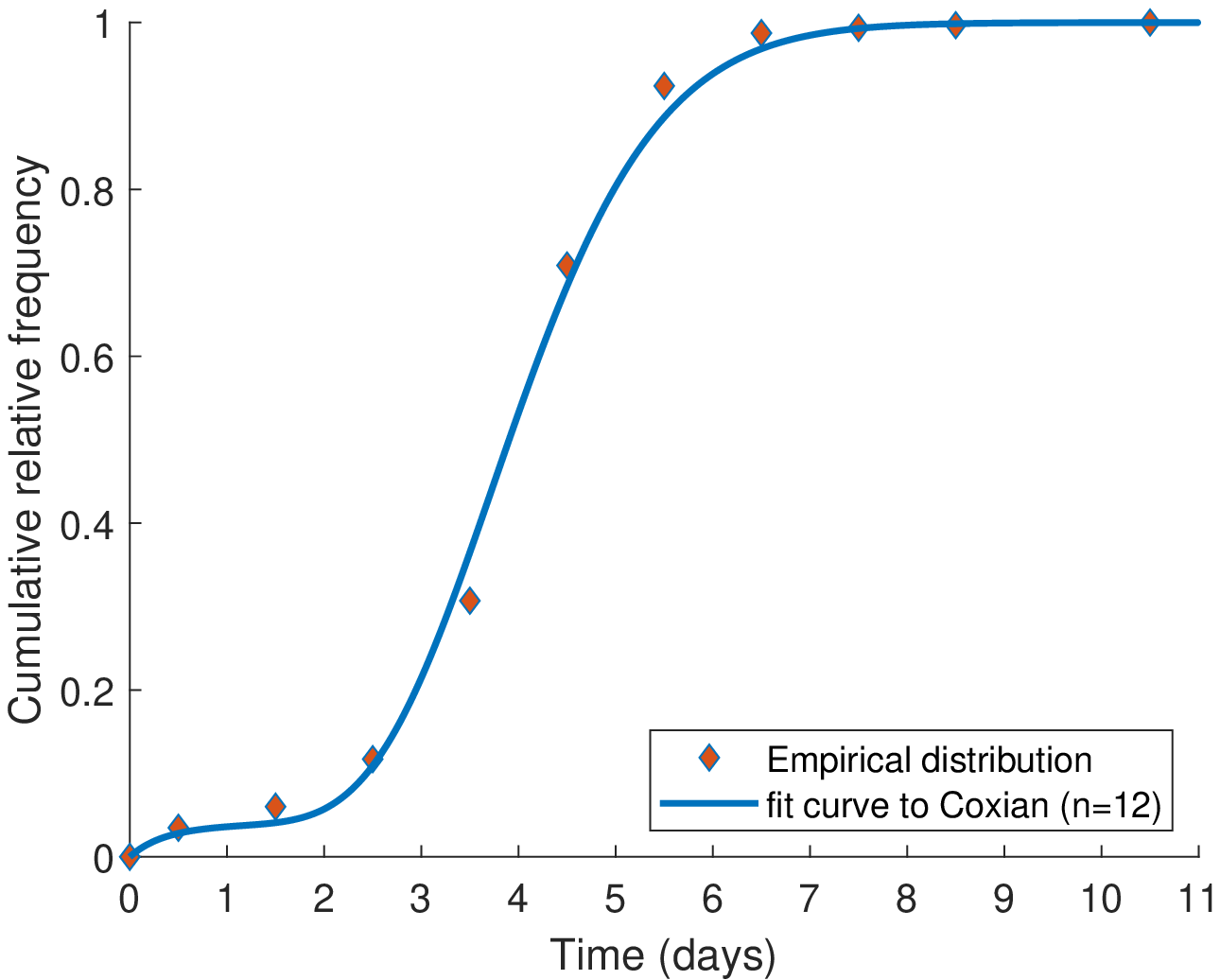}}~~~~~
\subfloat[]{\label{fig2b}\includegraphics[width=.45\textwidth]{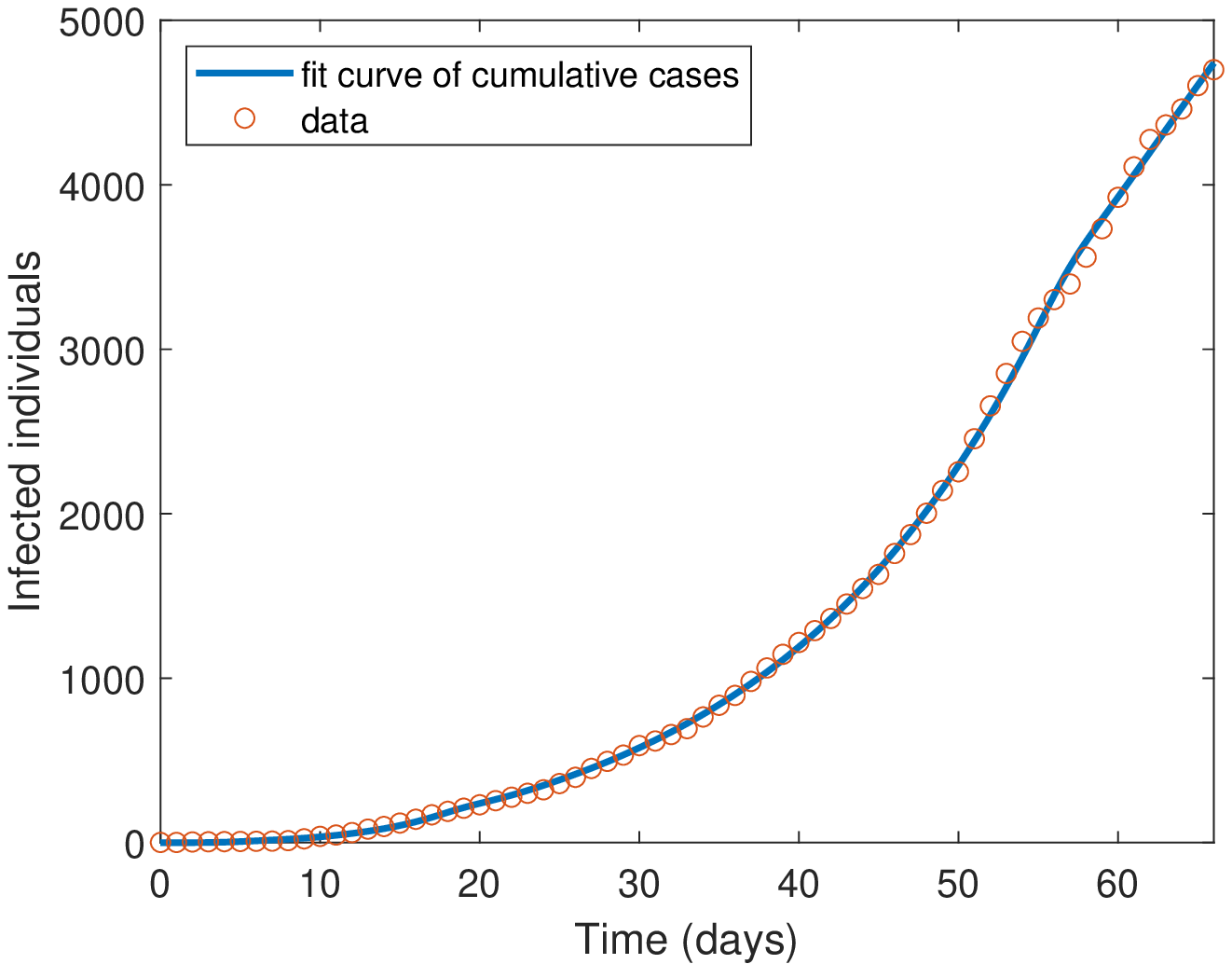}}
\caption{(a) Fitting empirical distribution of incubation period to Coxian of 12-chains, and (b) fitting case data to Coxian distributed SEIR model with 12-chains.}
\label{fig2}
\end{figure}

\newpage 
\section*{Tables}
\begin{table}[h!]
\tiny 
\begin{tabular}{@{}llll@{}}
\toprule
\begin{tabular}[c]{@{}l@{}}{\bf Distribution,}\\ { \bf Survival function($P_\triangle (t)$)}\end{tabular} & {\bf Derived model ($S(t)$, $I(t)$)}                                                                                                                                                                                                                                                                                                                                                                                                                       & \begin{tabular}[c]{@{}l@{}}{\bf  Basic reproduction}\\ { \bf number ($\beta(t)\equiv\bar{\beta}$)}\end{tabular} & {\bf Characteristics}                                                                                                                                                  \\ \midrule
 \begin{tabular}[c]{@{}l@{}}{\bf Dirac-delta}\\ $P_D(t)=\begin{cases} 1,~~u\in[0,\tau),\\ 0 ,~~u \in [\tau,\infty), \end{cases}$\\ for $\tau>0$, constant\end{tabular}                                                     & \begin{tabular}[c]{@{}l@{}}$\cfrac{\du E(t)}{\du t}=\beta(t) S(t)I(t) - \beta(t-\tau) S(t-\tau)I(t-\tau)\exp(-\mu \tau) -\mu E(t),$\\ $\cfrac{\du I(t)}{\du t} = \beta(t-\tau)S(t-\tau)I(t-\tau)\exp(-\mu \tau) -(\gamma+\mu) I(t).$\end{tabular}                                                                                                                                                                                                     & $\cfrac{\bar{\beta}\exp(-\mu\tau)}{\gamma+\mu}$                                & \begin{tabular}[c]{@{}l@{}}System of DDE,\\ Exact time of "delay", \\ NOT distributed \end{tabular}                                                \\ \midrule
\begin{tabular}[c]{@{}l@{}}{\bf Exponential}\\ $P_E(t)=\exp(-\lambda t)$                                                                                                                                                                            \end{tabular}                                                     & \begin{tabular}[c]{@{}l@{}}$\cfrac{\du E(t)}{\du t}=\beta(t) S(t)I(t) - (\lambda+\mu)E(t),$\\ $\cfrac{\du I(t)}{\du t}=\lambda E(t)-(\gamma+\mu)I(t).$\end{tabular}                                                                                                                                                                                                                                                                                   & $\cfrac{\bar{\beta}\mu}{(\lambda+\mu)(\gamma+\mu)}$                                 & \begin{tabular}[c]{@{}l@{}}System of simple ODE,\\ Memoryless property\end{tabular}                                                         \\ \midrule
  \begin{tabular}[c]{@{}l@{}}{\bf Gamma}\\ $P_G(u)=\sum_{i=0}^{n-1} \cfrac{1}{i!} e^{-n\lambda u}(n\lambda u)^i$,\\ for $n$, positive integer\end{tabular}                                                       & \begin{tabular}[c]{@{}l@{}}$\cfrac{\du E_n(t)}{\du t}= \beta(t) S(t)I(t)-(n\lambda+\mu)E_n, $\\ $\cfrac{\du E_i(t)}{\du t}= n\lambda E_{i+1}-(n\lambda+\mu)E_i,\textrm{ for }i=n-1,\cdots,2,1,$\\ $\cfrac{\du I(t)}{\du t}=n\lambda E_{1}(t) -(\gamma+\mu) I(t),$\\ with $E(t)=\sum_{i=0}^{n-1} E_{n-i}(t)$.\end{tabular}                                                                                                                             & $\cfrac{\bar{\beta}}{\gamma+\mu}\left(\cfrac{n\lambda}{n\lambda+\mu}\right)^n$                                & \begin{tabular}[c]{@{}l@{}}Linear chain,\\ unimodal, short-tailed \\ distribution \end{tabular}                                           \\ \midrule
  \begin{tabular}[c]{@{}l@{}}{\bf Mittag-Leffler} \\ $P_M(t)=E_{\alpha,1}\left(-(t/\zeta)^\alpha\right),$\\ for $0<\alpha\le 1$ \\ where $E_{\alpha,\beta}(z)=\sum_{k=0}^\infty \frac{z^k}{\Gamma(\alpha k+\beta)}$\end{tabular} & \begin{tabular}[c]{@{}l@{}}$\cfrac{\du E(t)}{\du t}=\beta(t) S(t)I(t) - \exp(-\mu t)\zeta^{-\alpha}{}_0\mathcal{D}_t^{1-\alpha}\left[\exp(\mu t)E(t)\right] -\mu E(t),$\\ $\cfrac{\du I(t)}{\du t}=\exp(-\mu t)\zeta^{-\alpha}{}_0\mathcal{D}_t^{1-\alpha}\left[\exp(\mu t)E(t)\right] -(\gamma+\mu) I(t),$\\ where ${}_0\mathcal{D}_t^{1-\alpha}[f(t)]=\cfrac{1}{\Gamma(\alpha)} \cfrac{\du }{\du t} \int_0^t (t-u)^{\alpha-1}f(u)\du u$.\end{tabular} & 	$\cfrac{\bar{\beta}}{\gamma+\mu} \cdot \cfrac{1}{1+(\zeta\mu)^\alpha}$                                 & \begin{tabular}[c]{@{}l@{}}System of FDE,\\ heavy-tailed distribution,\\ Hard to get exact form \\ of distribution\end{tabular} \\ \bottomrule
\end{tabular}
\caption{\label{tablehistory}SEIR models generated by various distributions with delay effect in the previous papers.}
\end{table}

\end{document}